\documentclass[journal]{IEEEtran}
\usepackage{amsmath,amsfonts}
\usepackage{algorithm}
\usepackage{array}
\usepackage[caption=false,font=normalsize,labelfont=sf,textfont=sf]{subfig}
\usepackage{textcomp}
\usepackage{stfloats}
\usepackage{url}
\usepackage{verbatim}
\usepackage{graphicx}

\usepackage{kotex}
\usepackage[T1]{fontenc}
\usepackage[utf8]{inputenc} 
\usepackage{xcolor}
\usepackage{algpseudocode}
\usepackage{booktabs}
\usepackage{multirow}

\newtheorem{theorem}{Theorem}

\newcommand{\bsH}{\boldsymbol{H}}
\newcommand{\ValImax}{50}

\begin{document}

\title{Hierarchical Subcode Ensemble Decoding of Polar Codes}

\author{Yubeen Jo, Geon Choi, Chanho Park,~\IEEEmembership{Student Member,~IEEE},  and Namyoon Lee,~\IEEEmembership{Senior Member,~IEEE}
        
\thanks{Yubeen Jo is with the School of Electrical Engineering, Korea University, Seoul 02841, Republic of Korea (e-mail: jybin00@korea.ac.kr).}
\thanks{Geon Choi, Chanho Park and Namyoon Lee are with the Department of Electrical Engineering, POSTECH, Pohang 37673, Gyeongbuk, Republic of Korea (e-mail: $\{$geon.choi, chanho26, nylee$\}$@postech.ac.kr).}
}

\maketitle
\begin{abstract}
 Subcode-ensemble decoders improve iterative decoding by running multiple decoders in parallel over carefully chosen subcodes, increasing the likelihood that at least one decoder avoids the dominant trapping structures. Achieving strong diversity gains, however, requires constructing many subcodes that satisfy a linear covering property—yet existing approaches lack a systematic way to scale the ensemble size while preserving this property. This paper introduces hierarchical subcode ensemble decoding (HSCED), a new ensemble decoding framework that expands the number of constituent decoders while still guaranteeing linear covering. The key idea is to recursively generate subcode parity constraints in a hierarchical structure so that coverage is maintained at every level, enabling large ensembles with controlled complexity. To demonstrate its effectiveness, we apply HSCED to belief propagation (BP) decoding of polar codes, where dense parity-check matrices induce severe stopping-set effects that limit conventional BP. Simulations confirm that HSCED delivers significant block-error-rate improvements over standard BP and conventional subcode-ensemble decoding under the same decoding-latency constraint.

 \end{abstract}


\section{Introduction}
To meet stringent requirements for ultra-reliable low-latency communication (URLLC) \cite{ 6g_channel_coding_trend, shtblk_len_ch_code, sparse_pretran_polar, deep_polar_code, boss_code, wsd}, ensemble decoding (ED) has emerged as an effective paradigm: multiple decoders are run in parallel to create decoding diversity, and the final decision is selected from the best candidate \cite{MBBP, AED_RM_Polar}. A central question is how to generate \emph{meaningful} diversity across the ensemble. Existing ED methods can be broadly grouped into two classes \cite{comparative_study_ED}: (i) approaches that keep the parity-check matrix (PCM) fixed and perturb the input or the decoding schedule—e.g., noise-aided ED (NED) \cite{NED_polar}, automorphism ensemble decoding (AED) \cite{AED_LDPC}, and scheduled ED (SED) \cite{comparative_study_ED}; and (ii) approaches that directly modify the PCM via row operations to produce multiple Tanner graphs, as in multiple-bases belief propagation (MBBP) \cite{MBBP}. For polar codes, AED-style methods have been explored by leveraging automorphisms inherited from Reed–Muller codes \cite{AED_RM_Polar}. However, it has been observed that, for polar codes, permutation-based diversity can largely vanish when the decoder is effectively invariant to such permutations \cite{AED_RM_Polar, AED_Polar}; in this case, altering the input does not translate into genuine graph diversity, and the ensemble gain becomes limited.

This motivates ED methods that create diversity by altering the PCM itself. While MBBP is conceptually appealing, constructing “good” alternative bases—often interpreted as seeking low-weight or well-conditioned representations—quickly leads to combinatorial complexity and even NP-hard subproblems, making it ill-suited for latency-critical operation \cite{SCED_LDPC}. A practical alternative is subcode ensemble decoding (SCED) \cite{SCED_LDPC}, which generates multiple decoders by augmenting a base PCM with additional linearly independent parity constraints. The effectiveness of SCED hinges on constructing subcodes that satisfy a linear covering (LC) property, so that the ensemble collectively provides broad coverage of the codeword space (e.g., via relations such as $\boldsymbol{h}_3=\boldsymbol{h}_1+\boldsymbol{h}_2$). The key limitation, however, is scalability: there is no systematic construction that increases the number of subcodes while \emph{guaranteeing} the LC property, and naïve extensions typically devolve into random search or incur rapidly growing complexity.

In this letter, we propose hierarchical subcode ensemble decoding (HSCED), a new ensemble construction that scales the number of constituent decoders while still guaranteeing linear covering. The key idea is to recursively construct subcode parity constraints in a hierarchical structure so that the LC property is preserved at every level, enabling large ensembles with controlled complexity. To enable a structured and implementation-friendly construction, we first transform the PCM into a sparse, upper-triangular basis via a row-reduced echelon form (RREF) preprocessing step, which preserves the original code space while exposing a convenient algebraic structure for recursive augmentation. We then build the ensemble by incrementally adding linearly independent parity constraints following the proposed hierarchy, yielding an explicit complexity–performance trade-off through the depth (and size) of the hierarchy.

We demonstrate the effectiveness of HSCED on belief propagation (BP) decoding of polar codes, where decoding is notoriously difficult due to the dense PCM structure and the resulting stopping-set behavior. Here, HSCED delivers a favorable asymmetric trade-off: although adding constraints can introduce some additional short cycles, the hierarchical construction systematically breaks dominant stopping sets, leading to a net performance gain. Simulation results confirm that HSCED achieves significant block-error-rate improvements over standard BP and conventional SCED under the same decoding-latency constraint, substantially narrowing the gap to stronger (but less parallelizable) decoding baselines.

\section{Preliminaries}

\subsection{Channel Coding and Polar Codes}

Consider a binary linear block code $\mathcal{C}(N, K)$ of blocklength $N$ and dimension $K$.
Let $\boldsymbol{u} \in \{0,1\}^K$ be an information vector consisting of independent and uniformly distributed bits.
The encoder maps $\boldsymbol{u}$ to a codeword $\boldsymbol{x} \in \{0,1\}^N$.
The code $\mathcal{C}$ is defined as the null space of a parity-check matrix (PCM) $\bsH \in \{0,1\}^{(N-K) \times N}$ that satisfies $\bsH \boldsymbol{x}^\top = \mathbf{0}$ for all $\boldsymbol{x} \in \mathcal{C}$.
Polar codes are constructed recursively using the kernel matrix $\boldsymbol{F} = \bigl[ \begin{smallmatrix} 1 & 0 \\ 1 & 1 \end{smallmatrix} \bigr]$ \cite{arikan-polar}.
For a code length $N=2^n$, the generator matrix is given by $\boldsymbol{G}_N = \boldsymbol{F}^{\otimes n}$.
The encoder input vector $\boldsymbol{v}$ is constructed by assigning $\boldsymbol{u}$ to the $K$ most reliable bit-channels denoted by $\mathcal{A}$ and setting the remaining frozen bits to zero.
The codeword is generated as $\boldsymbol{x} = \boldsymbol{v}\boldsymbol{G}_N$.
The standard polar PCM $\bsH_{\text{polar}}$ corresponds to the submatrix of $\boldsymbol{G}_N$ associated with the frozen indices.

\subsection{BP Decoding and Structural Weaknesses}

BP decoding is widely adopted for its low complexity and parallelism.
However, its performance on dense graphs is strictly limited by structural vulnerabilities such as short cycles and SS \cite{ldpc, Urbanke}.
Fig. \ref{fig:ss_four_cycle} illustrates these weaknesses in the Tanner graph.
A 4-cycle shown in blue is a closed path of length 4 that induces high correlation between messages and hinders convergence.
Furthermore, an SS is defined as a subset of variable nodes where every neighboring check node is connected at least twice within the subset.
As depicted in the red region of Fig. \ref{fig:ss_four_cycle}, check nodes connected to an SS cannot provide extrinsic information if all variable nodes in the set are erased or unreliable.
Consequently, small SSs are the primary cause of the error floor in BP decoding.

\begin{figure}[htbp!]
    \centering
    \includegraphics[width=0.95\linewidth]{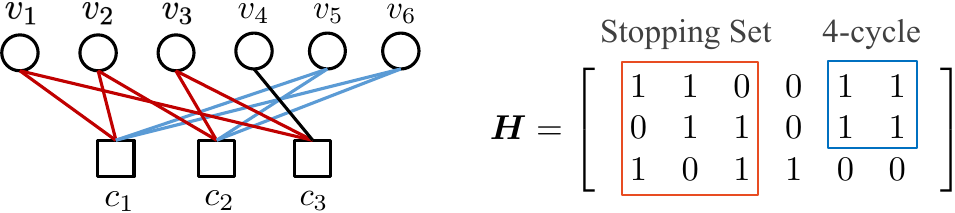}
    \caption{Illustration of structural weaknesses in PCM $\boldsymbol{H}$ where a 4-cycle (blue) represents the smallest stopping set and a larger stopping set is shown in red.}
    \label{fig:ss_four_cycle}
    \vspace{-2mm}
\end{figure}

\subsection{Subcode Ensemble Decoding}
To bridge the performance gap between BP and maximum likelihood (ML) decoding while maintaining low latency, SCED was proposed as a framework that generates graph diversity \cite{SCED_LDPC}.
Unlike AED which requires knowledge of the code's algebraic structure or MBBP which involves an NP-hard search for optimal bases, SCED offers a practical approach to construct diverse parity-check matrices without such constraints \cite{SCED_LDPC}.
The core principle is to construct an ensemble of subcodes that collectively cover the original code space $\mathcal{C}$.
A subcode $\mathcal{C}' \subset \mathcal{C}$ is defined by augmenting the original PCM $\bsH$ with an additional linearly independent row vector $\boldsymbol{h}$.
For notational convenience, we denote the vertical concatenation of $\bsH$ and $\boldsymbol{h}$ as $[\bsH; \boldsymbol{h}]$, yielding the new PCM:
\begin{align}
    \bsH_1 = [\bsH; \boldsymbol{h}_1] = \begin{bmatrix}
        \bsH \\ \boldsymbol{h}_1
    \end{bmatrix}.
\end{align}
This augmentation modifies the graph topology to provide a different decoding perspective.
To ensure that the ensemble can decode any valid codeword $\boldsymbol{x} \in \mathcal{C}$, the set of subcodes must satisfy the linear covering (LC) property $\bigcup \mathcal{C}_i = \mathcal{C}$.
It has been established that appending row vectors $\boldsymbol{h}_1$, $\boldsymbol{h}_2$, and $\boldsymbol{h}_3 = \boldsymbol{h}_1 + \boldsymbol{h}_2$ to the base PCM generates three subcodes whose union satisfies the LC property \cite{SCED_LDPC}.
In this work, we leverage this SCED framework to construct a hierarchical ensemble tailored for polar codes.

\section{Hierarchical Subcode Ensemble Decoding } \label{sec:HSCED}
The proposed HSCED framework provides a scalable decoding architecture by organizing subcodes into a tree-structured ensemble. Unlike conventional ensemble methods \cite{SCED_LDPC} that draw subcodes from a flat, unstructured pool, HSCED recursively expands the ensemble with depth, yielding progressively finer coverage of the code space while preserving the desired covering property at each level. The overall framework has two components: (i) hierarchical subcode construction and (ii) parallel ensemble decoding.

\subsection{Recursive Expansion via Linear Covering}
The core innovation of HSCED lies in its recursive construction. We define the ensemble at depth $d=0$ as the single base graph. From this root, we recursively expand the ensemble by applying the  LC property at each node. As illustrated in Fig. \ref{fig:subcode_tree}, a parent PCM at depth $d-1$ generates three child PCMs at depth $d$ by appending specific row vectors. Let $\mathcal{H}_d$ be the ensemble of PCMs at depth-$d$.

For example, to transition from depth 0 to depth 1, we generate three row vectors $\boldsymbol{h}_1, \boldsymbol{h}_2, \boldsymbol{h}_3$ satisfying $\boldsymbol{h}_3 = \boldsymbol{h}_1 + \boldsymbol{h}_2$.
Adding these to the base graph yields three subcodes defined by $\bsH_k = [\bsH; \boldsymbol{h}_k]$ for $k \in \{1, 2, 3\}$. 
Crucially, this process is repeated for each child node.
A subcode at depth 1 is further expanded into three subcodes at depth 2 using a new set of row vectors. 
Consequently, a subcode at depth $m$ accumulates $m$ additional constraints along its path from the root.
This hierarchical accumulation allows the decoder to target increasingly specific regions of the code space.

\begin{figure}[htbp!]
 \centering
 \includegraphics[width=\columnwidth]{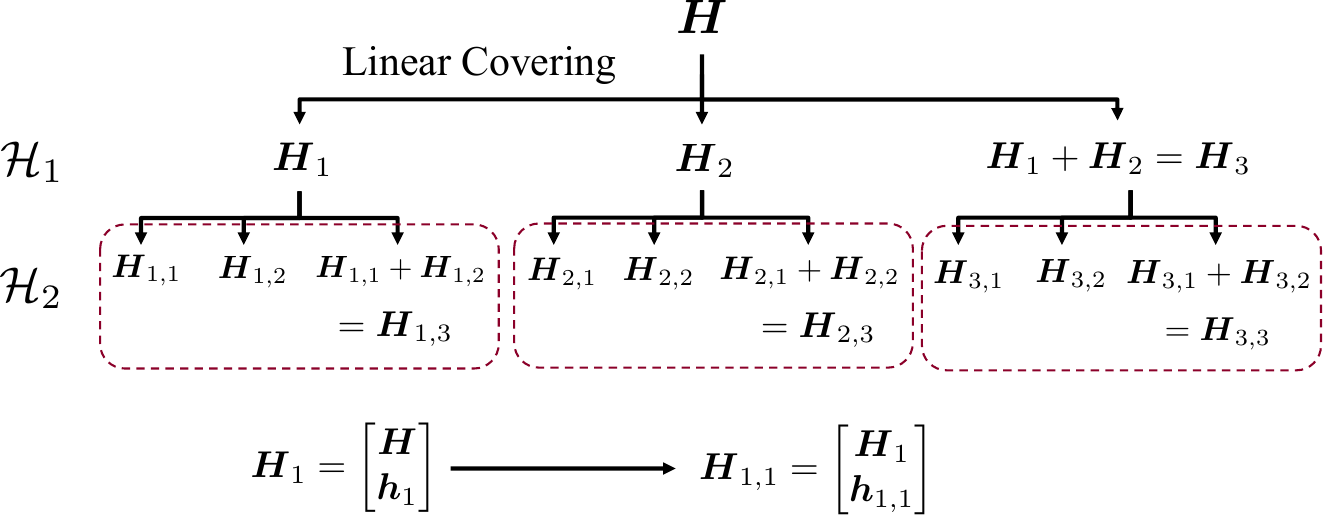}
 \caption{Hierarchical structure of the proposed HSCED ensemble. $\bsH_{i_1,\dots,i_m}$ denotes the PCM of a subcode generated by augmenting the base PCM with row vectors.}
 \label{fig:subcode_tree}
\end{figure}

To maintain the sparsity of the graph during this expansion, we employ a density-aware sampling strategy (Algorithm \ref{alg:hsced_construction}).
Let $p$ be the average density of the base graph.
For each expansion step, we construct row vectors using three disjoint components $\boldsymbol{h}_a, \boldsymbol{h}_b, \boldsymbol{h}_c$ each having a density of $p/2$.
By forming the added rows as $\boldsymbol{h}_1=\boldsymbol{h}_a+\boldsymbol{h}_c$, $\boldsymbol{h}_2=\boldsymbol{h}_b+\boldsymbol{h}_c$, and $\boldsymbol{h}_3 = \boldsymbol{h}_a + \boldsymbol{h}_b$, we ensure that the density of the graph remains controlled at approximately $p$ throughout the hierarchy.

\begin{theorem}[Hierarchical Linear Covering] \label{thm:subcode_size}
Let $\mathcal{C} = \text{Null}(\bsH)$. The union of subcodes defined by the ensemble ${\mathcal{H}}_m$ of $3^m$ PCMs (indexed by $\mathbf{i} \in \{1,2,3\}^m$) forms a complete linear covering of $\mathcal{C}$:
\begin{align}
\bigcup_{\mathbf{i} \in \{1,2,3\}^m} \text{Null}(\bsH_{\mathbf{i}}) = \mathcal{C}.
\end{align}
\end{theorem}

\begin{IEEEproof}
We proceed by induction on depth $m$.

\textbf{Base Case ($m=1$):} Let $\boldsymbol{h}_3 = \boldsymbol{h}_1 + \boldsymbol{h}_2$. As shown in \cite{SCED_LDPC}, the subcodes $\mathcal{C}_k = \text{Null}([\bsH; \boldsymbol{h}_k])$ for $k \in \{1,2,3\}$ cover $\mathcal{C}$ since any $\boldsymbol{x} \in \mathcal{C}$ failing parity checks for $\boldsymbol{h}_1$ and $\boldsymbol{h}_2$ must satisfy $\boldsymbol{h}_3$. Thus, $\bigcup_{k=1}^3 \text{Null}(\bsH_k) = \mathcal{C}$.

\textbf{Inductive Step:} Assume the hypothesis holds for depth $m-1$ denoted by index path $\mathbf{p} \in \{1,2,3\}^{m-1}$. For each parent code $\mathcal{C}_{\mathbf{p}} = \text{Null}(\bsH_{\mathbf{p}})$, applying the base case construction yields $\bigcup_{j=1}^3 \text{Null}(\bsH_{\mathbf{p}, j}) = \mathcal{C}_{\mathbf{p}}$.
Substituting this into the induction hypothesis, we obtain
\begin{align}
\bigcup_{\mathbf{p}} \left( \bigcup_{j=1}^3 \text{Null}(\bsH_{\mathbf{p}, j}) \right) = \bigcup_{\mathbf{p}} \text{Null}(\bsH_{\mathbf{p}}) = \mathcal{C}.
\end{align}
This completes the proof.
\end{IEEEproof}

\subsection{Parallel BP Ensemble Decoding}
The decoding process leverages the constructed hierarchical ensemble to achieve low latency.
As shown in Fig. \ref{fig:decoder_architecture}, the proposed HSCED decoder consists of a total of $3^m + 1$ parallel BP decoders.
These include one decoder for the base graph and $3^m$ decoders corresponding to the subcodes at depth $m$ (leaf nodes).
The received vector $\boldsymbol y$ serves as the common input for all decoders. All valid codewords found by the ensemble are collected into a candidate list $\mathcal{L}$.
Finally, the decoder selects the optimal codeword from $\mathcal{L}$ that minimizes the Euclidean distance to the received signal $\boldsymbol{y}$, satisfying the ML criterion ({\it{ML-in-the-list}}) \cite{AED_RM_Polar}.

\begin{figure}[htbp]
    \centering
    \includegraphics[width=0.9\columnwidth]{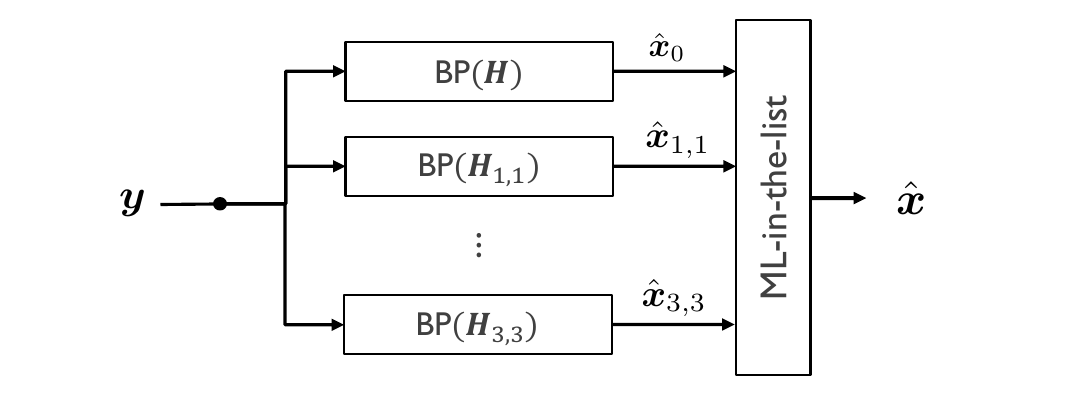}
    \caption{Architecture of the proposed HSCED decoder. At depth $m$, a total of $3^m+1$ parallel BP decoders operate on distinct subcodes.}
    \label{fig:decoder_architecture}
\end{figure}

\begin{algorithm}[tbhp!]
\caption{HSCED Ensemble Construction}
\label{alg:hsced_construction}
\begin{algorithmic}[1]
\small
\State \textbf{Input:} $\bsH$ (Base PCM), $m$ (Target depth)
\State \textbf{Output:} ${\mathcal{H}}_{m}$ ($3^m+1$ PCMs)
\State $p \leftarrow \text{density of } \bsH$ \Comment{Calculate base density}
\State ${\mathcal{H}}_{0} \leftarrow \{\bsH\}$
\For{$d = 1$ \textbf{to} $m$}
    \State $\tilde{\mathcal{H}}_{d} \leftarrow \emptyset$  \Comment{Initialize next level ensemble}
    \ForAll{$\bsH_{\text{parent}} \in {\mathcal{H}}_{d-1}$}
        \State $\boldsymbol{h}_a, \boldsymbol{h}_b, \boldsymbol{h}_c \leftarrow \text{SampleRowVectors}(N, p/2)$ 
        \State $\boldsymbol{h}_1 \leftarrow \text{mod}(\boldsymbol{h}_a + \boldsymbol{h}_c, 2)$ 
        \State $\boldsymbol{h}_2 \leftarrow \text{mod}(\boldsymbol{h}_b + \boldsymbol{h}_c, 2)$ 
        \State $\boldsymbol{h}_3 \leftarrow \text{mod}(\boldsymbol{h}_a + \boldsymbol{h}_b, 2)$
        \For{$i =1$ \textbf{to} $3$}
            \State $\bsH_i \leftarrow [\bsH_{\text{parent}}; \boldsymbol{h}_i]$
        \EndFor
        \State ${\mathcal{H}}_{d} \leftarrow {\mathcal{H}}_{d} \cup \{\bsH_1, \bsH_2, \bsH_3\}$
    \EndFor
\EndFor
\State \textbf{return} ${\mathcal{H}}_{m}$ \Comment{Return the leaf-node ensemble}
\end{algorithmic}
\end{algorithm}

\section{Application to Polar Codes} \label{sec:analysis}

To validate the structural advantages of the proposed framework, we analyze the evolution of graph properties when applying HSCED to standard 5G NR polar codes \cite{3gpp2020}.
We specifically focus on blocklength and dimensions $(N, K) \in \{(64, 32), (128, 96), (512, 464)\}$.
In this section, we examine how the sparsity, cycle distribution, and stopping sets of the graph change during the base graph generation and the subsequent hierarchical expansion.

Since the standard polar PCM $\bsH_{\text{polar}}$ inherently exhibits a dense structure which is ill-suited for BP decoding, we first establish a BP-friendly base graph.
Applying Gaussian elimination to $\bsH_{\text{polar}}$ yields the sparse upper-triangular matrix $\tilde{\bsH} = \mathsf{RREF}(\bsH_{\text{polar}})$ as shown in Fig. \ref{fig:RREF_H}.
This operation transforms the graph into a structure amenable to BP decoding while preserving the original code space.
We utilize this $\tilde{\bsH}$ as the deterministic base graph for the hierarchical ensemble construction described in Section \ref{sec:HSCED}.

\begin{figure}[htbp]
 \centering
 \includegraphics[width=\columnwidth]{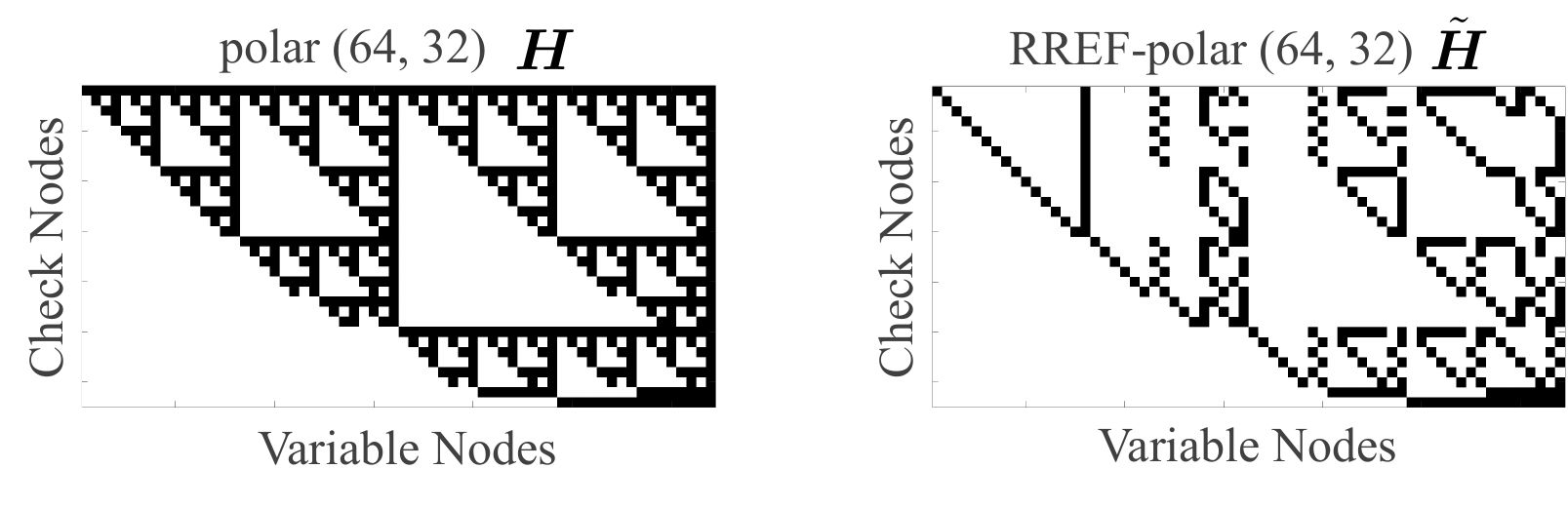}
 \caption{Heatmap comparison of the polar (64, 32) PCM before ($\bsH$) and after RREF application ($\tilde{\bsH}$).}
 \label{fig:RREF_H}
\end{figure}

Table \ref{tab:rref_effect} summarizes the structural improvements achieved by this transformation.
This confirms that RREF provides a favorable baseline for initiating the hierarchical ensemble construction.

\begin{table}[htbp]
\centering
\caption{Comparison of graph properties for polar codes before and after RREF.}
\label{tab:rref_effect}
\begin{tabular*}{\columnwidth}{l@{\extracolsep{\fill}}crrrr}
\toprule
$(N, K)$ & PCM & density(\%) & 4-cycle & $s=3$ SS & $s=4$ SS \\
\midrule
\multirow{2}{*}{$(64, 32)$} & $\bsH$ & $28.13$ & 16,690 & $0$ &$223$ \\
 & $\tilde{\bsH}$ & $\textbf{15.72}$ & $\textbf{2,036}$ & $0$ & $\textbf{27}$ \\
\midrule
\multirow{2}{*}{$(128, 96)$} & $\bsH$ & $21.58$ & 83,674 & $80$ & 7,458\\
 & $\tilde{\bsH}$ & $\textbf{12.01}$ & $\textbf{16,524}$ & $\textbf{37}$ & $\textbf{924}$ \\
\midrule
\multirow{2}{*}{$(512, 464)$} & $\bsH$ & $28.39$ & 2,330,700 & 4,008 & $N/A$ \\
 & $\tilde{\bsH}$ & $\textbf{19.14}$ & $\textbf{483,824}$ & $\textbf{1,438}$ & $N/A$ \\
\bottomrule
\end{tabular*}
\parbox{\columnwidth}{\footnotesize $N/A$: Analysis omitted due to computational complexity.}
\vspace{-4mm}
\end{table}

\subsection{Structural Effects of HSCED} \label{sec:HSCED_Mechanism}

We now analyze the impact of the hierarchical row augmentation on the graph structure.
Although $\tilde{\bsH}$ is sparser than $\bsH$, residual SSs persist due to the graph's structural characteristics.
The efficacy of HSCED stems from the specific topological properties of the RREF-transformed graph which exhibits an \textit{asymmetric gain} between SS reduction and cycle growth.
The hierarchically added row $\boldsymbol{h}$ functions as a global constraint that establishes connections between variable nodes across distant indices.
Empirically, we observed that setting $p$ to the average density of $\tilde{\bsH}$ yields the best performance.

Quantitative analysis in Fig. \ref{fig:4cycle_analysis} and Table \ref{tab:hsced_tradeoff} confirms this mechanism.
As shown in Table \ref{tab:hsced_tradeoff}, at depth $m=4$ for the $(64, 32)$ code, the dominant SS ($s=4$) are suppressed from 27 to 4.1 on average, representing an $85\%$ reduction.
Similarly, for the $(128, 96)$ code, the $s=4$ SSs are reduced by approximately $88\%$ from 924 to 110.4.
In contrast, the increase in 4-cycles is marginal, specifically only $8.6\%$ for $N=64$.
This result demonstrates that the proposed augmentation strategy efficiently eliminates structural weaknesses with negligible degradation in graph sparsity.
This gain is the key factor enabling HSCED to outperform conventional methods on polar codes.

\begin{figure}[htbp]
 \centering
 \includegraphics[width=1\columnwidth]{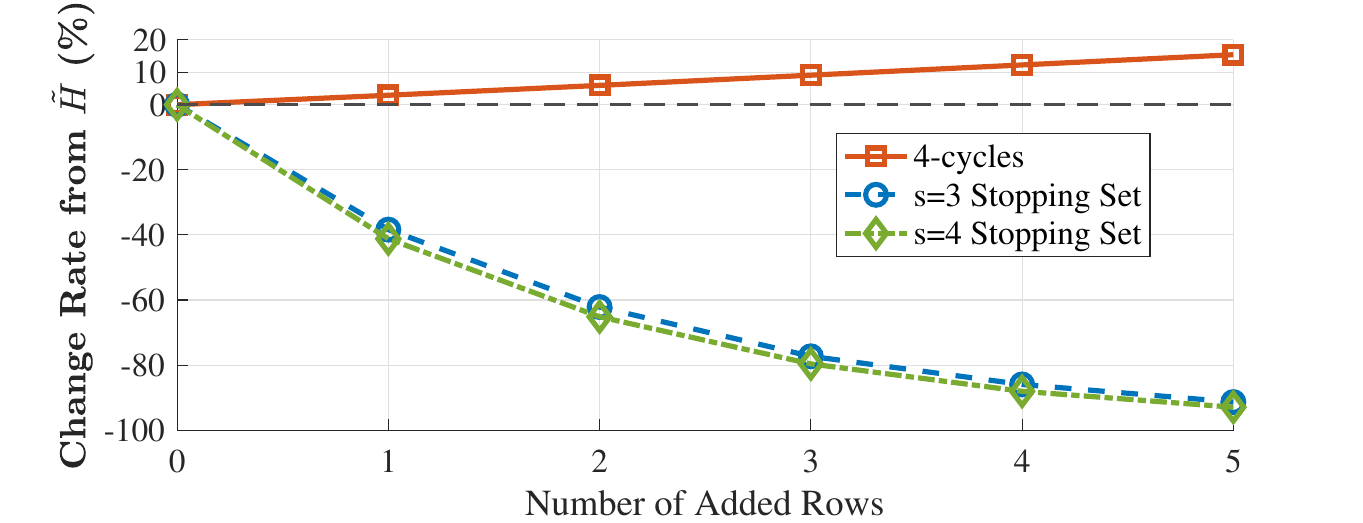}
 \caption{Asymmetry of gain (stopping set reduction) versus cost (4-cycle increase) when adding rows to the RREF-Polar $(64, 32)$ graph $T(\tilde{\bsH})$.}
 \label{fig:4cycle_analysis}
 \vspace{-2mm}
\end{figure}

\begin{table}[htb]
\centering
\caption{Percentage change in graph properties from $\mathsf{RREF}(\bsH_\text{polar})$ with HSCED application (avg. over $2,000$ trials).}
\label{tab:hsced_tradeoff}
\begin{tabular*}{\columnwidth}{@{\extracolsep{\fill}}llrrrr}
\toprule
$(N, K)$ & PCM & 4-cycle & $s=3$ & $s=4$ &$s=5$ \\
\midrule
\multirow{2}{*}{$(64, 32)$} & $\tilde\bsH$ & 2,036 & $0$ &$27$ &$530$ \\
 & $\tilde{\bsH}_{i_1,i_2,i_3,i_4}$  & 2,212 & $0$ & $\textbf{4.1}$ & $\textbf {69.8}$\\
\midrule
\multirow{2}{*}{$(128, 96)$} & $\tilde\bsH$  & 16,524 & $37$ & $924$ & $N/A$\\
 & $\tilde{\bsH}_{i_1,i_2,i_3,i_4}$ & 18,541 & $\textbf{5.2}$ & $\textbf{110.4}$ & $N/A$ \\
\midrule
\multirow{2}{*}{$(512, 464)$} & $\tilde\bsH$ & 483,824 & 1,438 & $N/A$  & $N/A$\\
 & $\tilde{\bsH}_{i_1,i_2,i_3,i_4}$ & 523,081 & $\textbf{218.4}$ & $N/A$ & $N/A$\\
\bottomrule
\end{tabular*}
\parbox{\columnwidth}{$N/A$: Analysis omitted due to computational complexity.}
\vspace{-6mm}
\end{table}

\section{Simulation Results}\label{sec:simulation}

In this section, we evaluate the block error rate (BLER) performance of the proposed HSCED scheme on the 5G NR polar codes \cite{3gpp2020} discussed in Section \ref{sec:analysis}.
We compare the decoding performance of HSCED against conventional BP decoding, the existing SCED method, and the SCL decoder to demonstrate the effectiveness of the hierarchical ensemble construction.
Furthermore, we analyze the trade-off between decoding complexity and latency to highlight the suitability of HSCED for URLLC applications.

\subsection{Simulation Settings and Benchmarks}
We evaluate the BLER performance on 5G NR polar codes with blocklength and dimensions $(N, K) \in \{(64, 32), (128, 96), (512, 464)\}$ over a binary-input additive white Gaussian noise (AWGN) channel. All PCMs are preprocessed using RREF. For BP decoding, we employ the normalized min-sum algorithm (MSA) with a normalization factor $\alpha=0.75$ to mitigate the overestimation of LLR magnitudes.

For the SCED benchmark, we adopted a rigorous selection strategy to ensure a fair comparison.
We first collected an error set consisting of $1,000$ erroneous codewords generated by the MSA decoder at a target BLER of $10^{-3}$.
Subsequently, we generated $5,000$ random candidate row vectors and evaluated them on this set of collected error.
The single best vector $\boldsymbol{h}$ that corrected the maximum number of errors was selected for the SCED ensemble.
In contrast, our proposed HSCED ensembles, specifically depth 3 (HSCED-27 with a total of 28 decoders) and depth 4 (HSCED-81 with a total of 82 decoders), were generated using the structured sampling strategy outlined in Algorithm \ref{alg:hsced_construction} without any such posterior selection or optimization. This setup highlights the robustness and efficiency of the proposed hierarchical construction.

For all BP decoders, the maximum iteration count is set to $I_{\max} = \ValImax$. When adding rows to $\tilde{\bsH}$, the target density $p$ corresponds to the average density of $\tilde{\bsH}$. As a serial decoding benchmark, the SCL decoder is implemented based on \cite{scl} with a list size of $L=32$.

\begin{figure*}[t!]
  \centering
  \includegraphics[width=\textwidth]{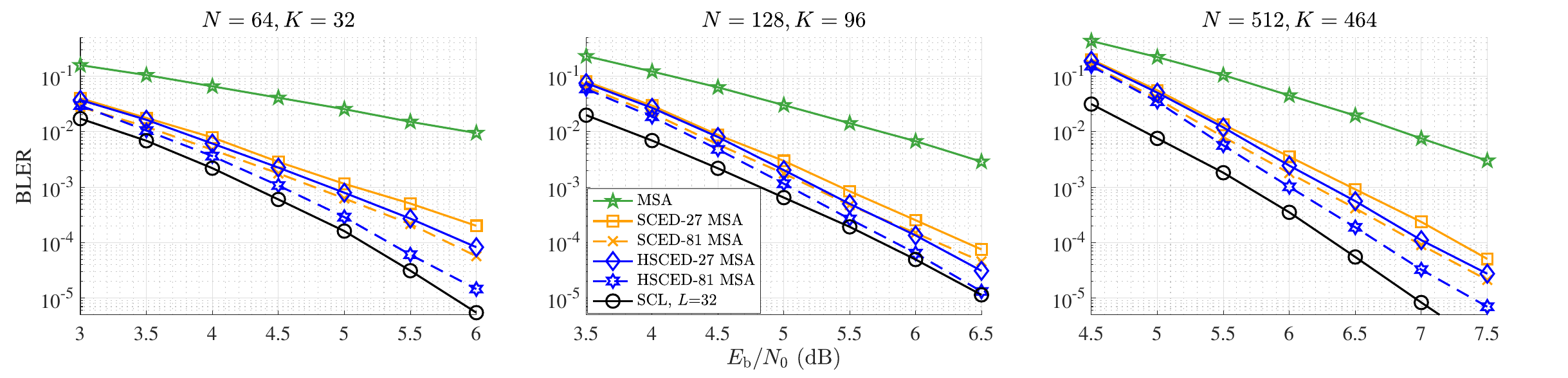}
  \caption{BLER performance comparison of MSA, SCED, HSCED, and SCL decoding for 5G NR polar codes with $(N, K) \in \{(64, 32), (128, 96), (512, 464)\}$. All BP-based decoders use $I_{\max}=\ValImax$ with early stopping.}
  \vspace{-2mm}
  \label{fig:polar_performance_combined}
\end{figure*}

\subsection{Performance, Complexity, and Latency Analysis}

Fig. \ref{fig:polar_performance_combined}
confirms that the proposed HSCED consistently outperforms both MSA and SCED, closely approaching the reliability of SCL32 at $N=128$. This empirical result validates our structural analysis that the systematic removal of dominant stopping sets effectively compensates for the modest increase in short cycles.

Table \ref{tab:complexity_comparison_op} and \ref{tab:latency_comparison_op} highlight the critical trade-off between computational cost and decoding latecny in the high $E_b/N_0$ region targeting URLLC service. HSCED requires a higher total operation count than SCL32 (e.g., $8\times$ operations at $N=512$) due to its parallel ensemble structure. However, this investment yields a decisive advantage in latency. 
While SCL suffers from a latency that scales linearly with blocklength ($2N-2$ cycles \cite{scl_2N_2}), HSCED leverages massive paralleism to ensure a strictly bounded worst-case latency of $2 I_{\max} = 100$ cycles regardless of $N$ (assuming $T_{\text {iter}}=2$ \cite{bp_latency}). Consequently, HSCED achieves a $10\times$ speedup over SCL32 at $N=512$. This demonstrates that HSCED serves as a specialized solution for URLLC scenarios, where trading hardware resources for ultra-low, bounded latency is a justifiable design choice.

\begin{figure}[tbh!]
  \centering
  \includegraphics[width=\columnwidth]{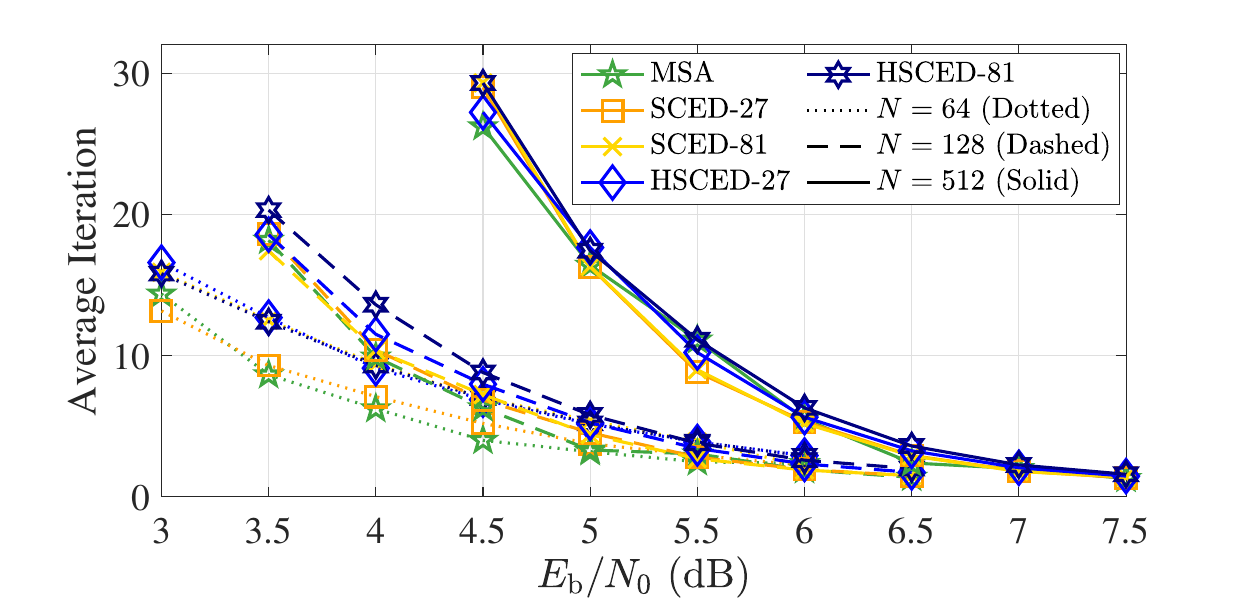}
  \caption{Iterations comparison of MSA, SCED, HSCED, and SCL decoding for 5G NR polar codes with $(N, K) \in \{(64, 32), (128, 96), (512, 464)\}$. All BP-based decoders use $I_{\max}=\ValImax$ with early stopping.}
  \label{fig:HSCED_performance}
  \vspace{-1mm}
\end{figure}

\begin{table}[t!]
\centering
\caption{Comparison of Average Computational Complexity (Total Operations) at High $E_b/N_0$}
\label{tab:complexity_comparison_op}
\resizebox{\columnwidth}{!}{%
\begin{tabular}{@{}lccccc@{}}
\toprule
\multirow{2}{*}{\textbf{Decoder}} & \multirow{2}{*}{\textbf{Complexity Formula}} & \multicolumn{3}{c}{\textbf{Total Operations ($\times 10^4$)}} \\ \cmidrule(l){3-5} 
 &  & \textbf{$N=64$} & \textbf{$N=128$} & \textbf{$N=512$} \\ \midrule
\textbf{SCL ($L=32$)} & $L \cdot N \cdot \log_2 N$ & \textbf{1.23} & \textbf{2.87} & \textbf{14.75} \\ \midrule
MSA (BP) & $2 \cdot |E| \cdot I_{\text{avg}}$ & 0.15 & 0.14 & 1.22 \\
SCED-27 & $M \cdot 2 \cdot |E| \cdot I_{\text{avg}}$ & 3.88 & 3.96 & 34.32 \\
SCED-81 & $M \cdot 2 \cdot |E| \cdot I_{\text{avg}}$ & 11.48 & 11.96 & 106.68 \\ \midrule
\textbf{HSCED-27 (Proposed)} & $M \cdot 2 \cdot |E| \cdot I_{\text{avg}}$ & 5.09 & 4.60 & 37.34 \\
\textbf{HSCED-81 (Proposed)} & $M \cdot 2 \cdot |E| \cdot I_{\text{avg}}$ & 15.13 & 15.94 & 121.93 \\ \bottomrule
\end{tabular}%
}
\vspace{1mm}
\footnotesize{\\Complexity values are calculated at the highest simulated $E_b/N_0$ for each block length ($6.0$ dB for $N=64$, $6.5$ dB for $N=128$, and $7.5$ dB for $N=512$).}
\vspace{-2mm}
\end{table}

\begin{table}[tbh!]
\centering
\caption{Comparison of Latency in Clock Cycles at High $E_b/N_0$ (Average / Worst-case)}
\label{tab:latency_comparison_op}
\resizebox{\columnwidth}{!}{%
\begin{tabular}{@{}lccccc@{}}
\toprule
\multirow{2}{*}{\textbf{Decoder}} & \multirow{2}{*}{\textbf{Latency Formula}} & \multicolumn{3}{c}{\textbf{Latency [Avg / Worst] (Cycles)}} \\ \cmidrule(l){3-5} 
 &  & \textbf{$N=64$} & \textbf{$N=128$} & \textbf{$N=512$} \\ \midrule
\textbf{SCL ($L=32$)} & $2N-2$ \cite{scl_2N_2} & 126 / 126 & 254 / 254 & 1022 / 1022 \\ \midrule
MSA (BP) & $I \times T_{\text{iter}}$ & 4.7 / 100 & 2.8 / 100 & 2.6 / 100 \\
SCED-27 & $I \times T_{\text{iter}}$ & 4.5 / 100 & 3.0 / 100 & 2.7 / 100 \\
SCED-81 & $I \times T_{\text{iter}}$ & 4.4 / 100 & 3.0 / 100 & 2.8 / 100 \\ \midrule
\textbf{HSCED-27 (Proposed)} & $I \times T_{\text{iter}}$ & 5.9 / 100 & 3.5 / 100 & 2.9 / 100 \\
\textbf{HSCED-81 (Proposed)} & $I \times T_{\text{iter}}$ & 5.6 / 100 & 4.0 / 100 & 3.2 / 100 \\ \bottomrule
\end{tabular}%
}
\vspace{1mm}
\footnotesize{\\Latency is measured in clock cycles assuming $T_{\text{iter}} = 2$ and $I_{\max} = 50$. Average latency is taken at the highest simulated $E_b/N_0$ ($6.0$, $6.5$, and $7.5$ dB for $N=64, 128, 512$, respectively).}
\vspace{-2mm}
\end{table}

\section{Conclusion}

This paper proposed HSCED, a scalable ensemble framework that extends the SCED principle while systematically scaling the number of subcodes under a guaranteed linear covering property. By constructing the ensemble on top of an RREF-based sparse base graph, HSCED achieves an asymmetric structural gain: the modest increase in short cycles is more than compensated by the systematic breaking of dominant stopping sets that limit BP decoding. Although HSCED increases aggregate computation as the ensemble grows, it preserves fixed low latency through parallelization, and simulations show clear block-error-rate improvements over standard BP and conventional SCED under the same latency constraint—making HSCED a strong candidate for URLLC regimes where near-SCL reliability is desired.

\bibliographystyle{IEEEtran}
\bibliography{letter}

\newpage

\vfill

\end{document}